\documentclass[letterpaper]{sig-alternate-per}

\usepackage{amsmath}
\usepackage{amsthm}
\usepackage{amssymb}
\usepackage{verbatim}
\usepackage{graphicx}

\usepackage{tikz}
\usetikzlibrary {positioning}
\usetikzlibrary{shapes,snakes}
\usetikzlibrary{arrows,calc}
\usepackage{relsize}
\usetikzlibrary{patterns}
 \usepackage{algorithm}
\usepackage{algpseudocode}

\theoremstyle{plain}
\newtheorem{lemma}{Lemma}
\newtheorem{rem}{Remark}
\newtheorem{theorem}{Theorem}

\theoremstyle{definition}

\floatname{algorithm}{Algorithm}

\allowdisplaybreaks[4]

\usepackage{amsfonts}
\usepackage{times}
\usepackage{latexsym}
\usepackage{amssymb}
\usepackage{amsmath}
\usepackage{cite}
\usepackage{verbatim}
\def\bb0{{\mathbb{0}}}

\def\bb{{\mathbf{b}}}

\def\b0{{\mathbf{0}}}

\def\b1{{\mathbf{1}}}

\def\bbE{{\mathbb{E}}}

\def\cA{\mathcal{A}}

\def\sfM{\mathsf{M}}

\def\sf0{{\mathsf{0}}}

\newcommand{\marceau}[1]{\color{black}#1\color{black}}

\begin{document}

\newlength{\figurewidth}\setlength{\figurewidth}{0.6\columnwidth}


\addtolength{\topmargin}{-0.5\baselineskip}
\addtolength{\textheight}{\baselineskip}

\title{\fontsize{23}{23}\selectfont Online Budgeted Truthful Matching}

\newcounter{one}
\setcounter{one}{1}
\newcounter{two}
\setcounter{two}{2}


\addtolength{\floatsep}{-\baselineskip}
\addtolength{\dblfloatsep}{-\baselineskip}
\addtolength{\textfloatsep}{-\baselineskip}
\addtolength{\dbltextfloatsep}{-\baselineskip}
\addtolength{\abovedisplayskip}{-1ex}
\addtolength{\belowdisplayskip}{-1ex}
\addtolength{\abovedisplayshortskip}{-1ex}
\addtolength{\belowdisplayshortskip}{-1ex}

\numberofauthors{2} 
%
\author{
%
%
\alignauthor
Rahul Vaze \\
       \affaddr{School of Technology and Computer Science,  Tata Institute of Fundamental Research, Mumbai, India}
\and
\alignauthor 
Marceau Coupechoux
 \\
 \affaddr{LTCI, CNRS, Telecom ParisTech, University Paris-Saclay, France}
}
\maketitle
\begin{abstract}
An online truthful budgeted matching problem is considered for a bipartite graph, where the right vertices are available ahead of time, and individual left vertices arrive sequentially. 
On arrival of a left vertex, its edge utilities (or weights) to all the right vertices and a corresponding cost (or bid) are revealed.  If a left vertex is matched to any of the right vertices, then it has to be paid at least as much as its cost. 
The problem is to match each left vertex instantaneously and irrevocably
to any one of the right vertices, if at all, to find the maximum weight matching that is truthful,
under a payment budget constraint.  Truthfulness condition requires that no left vertex has any incentive of misreporting its cost. Assuming that the vertices arrive in an uniformly random order (secretary model) with arbitrary utilities, a truthful algorithm is proposed that is $24\beta$-competitive (where $\beta$ is the ratio of the maximum and the minimum utility) and satisfies the payment budget constraint. 
Direct applications of this problem include crowdsourcing auctions, and matching wireless users to cooperative relays in device-to-device enabled cellular network. 
\end{abstract}
\vspace{-0.14in}
\section{Introduction} \label{sec:intro}
Motivated by applications in crowdsourcing and device-to-device (D2D) communication, we consider an online bipartite matching problem over a bipartite graph $G(L\cup R,E)$, where the right vertex set $R$ is known ahead of time, while left vertices of $L$ arrive sequentially in a random order. The incident edge utilities from a vertex $\ell\in L$ to set $R$ are revealed only upon its arrival, as well as its cost $c_{\ell}$, and the problem is to decide which vertex of $R$ to match with $\ell$, if at all, immediately and irrevocably. If vertex $\ell$ is matched, a payment $p_{\ell}$ is made to vertex $\ell$ that has to be at least as much as its reported cost $c_{\ell}$.  A total  budget constraint of $B$ is assumed for payments to be made to the matched left vertices. We assume that left vertices are strategic players, which could potentially manipulate the reporting of their true cost, and hence seek a truthful algorithm, i.e., such that no incoming vertex has incentive to misreport its cost. 

The considered problem is a {\it truthful} generalization of the online knapsack problem, where each item has a value and a weight, and the problem is to accept each item instantaneously and irrevocably, so as to maximize the sum of the values of the accepted items subject to a sum-weight constraint. If all left vertices in our problem are truthful, then keeping $p_{\ell} = c_{\ell}$ satisfies the truthfulness constraint, and the considered problem specializes to the knapsack problem with knapsack size of $B$. The best known algorithm for the online knapsack problem is a $10e$-competitive algorithm \cite{babaioff2007knapsack}. 
\vspace{0.3in}

The considered problem is a special case of a reverse auction \cite{myerson1981optimal}, where users submit bids for accomplishing a set of tasks and if selected, expect a payment at least as much as their reported bids. The generality is in not enforcing the one-to-one matching constraint, i.e., one user can do more than one task, or any task can be assigned to more than one user.
 Budget feasible mechanisms have been introduced in \cite{Singer10} for reverse auctions (set $R$ is made of a single vertex). For the offline problem, with a matching constraint similar to this paper, a $3$-approximate algorithm has been derived in \cite{goel2013matching} that is one-sided truthful. When  the goal is to maximize the number of assigned tasks, \cite{Singer13} provides a $320$-competitive truthful algorithm assuming the secretary-input model. 
 
Some important applications of the considered problem are in crowdsourcing and device-to-device (D2D) communication, that is expected to be part of modern wireless communication standards. For crowdsourcing applications, a platform advertises a set of tasks it wants to accomplish and multiple users successively bid for completing those tasks and expect some payment towards that end. The job of the platform is to select the set of users that maximize its utility under a total budget constraint. If the platform is not careful in selecting its payment strategy, bidders have an incentive to misreport their costs \cite{yang2012crowdsourcing, subramanian2015online}, and therefore we seek a truthful algorithm for matching users to tasks and decide on corresponding payments.

From a D2D communication perspective, consider a single base station (BS) and a set of users $U$ that are connected to that BS. Let $R\subseteq U$ be the set of active users, while the remaining $L= U\backslash R$ are inactive but can potentially help users of set $R$ to relay their communication to/from the BS, as envisaged in future networks. For any inactive user $\ell \in L$, the set of active users it can help is $R(\ell) \subseteq R$ and at any time it can only help any one user from $R(\ell)$. Since relaying requires $\ell$ to spend some of its resources, e.g., battery, it is natural to assume that $\ell$ expects some payment for its help, and submits a corresponding bid at the time of advertising its inactive state and ability to help. The job of the BS is to allocate at most one helper (matching) from set $L$ for each user in set $R$, and decide the corresponding payment that is at least as much as the bid of that user in $L$. 
Clearly, there is incentive for users in set $L$ to misreport their bids in order to extract more payments and this motivates the need to seek truthful matching algorithm.

The main contribution of this paper is a $24\beta$-competitive randomized matching algorithm that is truthful and satisfies the payment budget constraint, where $\beta$ is the ratio of the largest to the smallest utility of any edge. To keep the problem non-degenerate, similar to other prior related works on online algorithms \cite{babaioff2007knapsack, KorulaPal}, we consider a secretarial input model, where the order of arrival of left vertices is uniformly random, but their utilities and bids can be arbitrary or even adversarial. Under this model, we modify the offline algorithm \cite{goel2013matching} and then use the sample and price algorithm \cite{babaioff2007knapsack} to make the algorithm online, where the novelty is in terms of defining the price for each right vertex and the corresponding payment for any left vertex that is matched to the right vertex. To contrast with the online knapsack problem, which we noted is a special case where truthfulness is guaranteed, the price of truthfulness is $24\beta/10e $ in terms of competitiveness.

\section{Online Truthful Budgeted \\ Matching}
Let $G = (L \cup R, E)$ be a bipartite graph with left vertices $L$, and right vertices $R$, and edge set $E$. 
The utility or weight of each edge $e$ is $u(e)$. For a set of edges $E'$, its utility is $u(E') = \sum_{e\in E'} u(e)$. 
Each left vertex $\ell$ has an associated cost or bid $c_{\ell}$, that does not depend on the right vertex $r$. We denote $c(e)=c_{\ell}$ as the cost/bid of edge $e=(\ell,r)$. 
If a left vertex $\ell$ is matched to a \marceau{right vertex $r$}, then a minimum payment of \marceau{$c_{\ell}$ } has to be made to user $\ell$, with an overall budget constraint of $B$. Let $u_{max} = \max_e u(e)$ and $u_{min} = \min_e u(e)$. We assume that $\frac{u_{max}}{u_{min}} \le \beta$ . Moreover, we also assume the typical large market assumption \cite{goel2013matching}, i.e., $\frac{u_{max}}{u^*}$ is small, where $u^*$ is the optimal sum-utility  of the matching under the budget constraint. Thus, no single user can influence the outcome significantly.

\begin{rem}As shown in \cite{yang2012crowdsourcing}, if bids of left vertices are used as payments, there is  incentive for left vertices to misreport their bids, and consequently the mechanism is not truthful or incentive compatible. Thus, the payment strategy is non-trivial. 
\end{rem} 

In this work, we consider the online problem, where the set $R$ of right vertices is known ahead of time, while the left vertices of set $L$ arrive sequentially in time and reveal their edge set (and utilities) incident on the right vertices together with their bid. 
On each left vertex arrival, it has to be matched irrevocably to any one of the unmatched right vertex at that time, if at all. If a left vertex is matched, then the payment to be made to it is also decided at the time of its arrival that cannot be changed later. 

To keep the problem non-degenerate in terms of competitive ratio, we assume that the order of arrival of left vertices is uniformly random, that is each permutation of left vertices is equally likely. As a result, the objective is to find a truthful algorithm with constant expected competitive ratio under the payment budget constraint of $B$. The weights (utilities), however, are allowed to be selected by an adversary.

Before considering the online scenario, we first discuss the offline (all left vertices and their edges are available non-causally) case, and define the optimal fractional solution (matching under budget constraint) to be $\mathsf{OPT}(B)$.
Note that for defining OPT, we are using raw bids as payments and truthfulness is not required. 
We note an important property for $\mathsf{OPT}(B)$ whose proof is immediate.
\begin{lemma}\label{lem:optscaling} For $\alpha \le 1$, 
$u(\mathsf{OPT}(B)) \le \frac{1}{\alpha}u(\mathsf{OPT}(\alpha B))$.
\end{lemma}

For the offline scenario, we propose a \textsc{Threshold} algorithm that is inspired (a modified version) by the \textsc{UniformMechanism} algorithm \cite{goel2013matching}, where the \textsc{Greedy} subroutine is the usual greedy matching algorithm for a bipartite graph. 
We define for each edge $e = (\ell, r)$ a buck per bang $b(e) = \frac{c(e)}{u(e)}$ that represents the \marceau{cost per unit utility}. For any $\gamma$, let $G(\gamma)$ be the graph obtained by removing all edges $e \in E(G)$ with buck per bang $b(e) > \gamma$.

\begin{algorithm}
\caption{\textsc{Threshold}}\label{alg:unimech}
\begin{algorithmic}[1]
\State {\bf Input:} Graph $G$, Budget $B$, $m=|E(G)|$
\State {\bf Output:} Matching $\sfM$, Threshold $\gamma_B$
	\State$\cA(G) = \{\gamma : \sum_{e\in \sfM}\gamma u(e)\leq B,\; \sfM=\mbox{\textsc{Greedy}}(G(\gamma))\}$
	\State $\gamma_B=\max\{\gamma: \gamma \in \cA(G) \}$
\State Accept all users in $\sfM = \mbox{\textsc{Greedy}}(G(\gamma_B))$ 
\end{algorithmic}
\end{algorithm}
The main idea behind the \textsc{Threshold} algorithm is to find 
the largest threshold $\gamma_B$, (subject to budget constraint), such that all edges whose buck per bang is more than that threshold are not considered for matching, while maintaining a (greedy) matching with large enough sum-utility. 
\begin{lemma}\label{lem:umguarantee}Let $\sfM$ be the matching output by \textsc{Threshold} algorithm with input graph $G$ under budget constraint $B$. Then $u(\sfM) \ge \frac{\mathsf{OPT}(B)}{3}$.
\end{lemma}
\begin{proof} 

For convenience we suppress the dependence on $B$ whenever its not essential to do so. Decompose the optimal fractional matching solution $\mathsf{OPT} = \{\mathsf{OPT}^{+} \cup \mathsf{OPT}^{-}\}$, where $\mathsf{OPT}^{+}$ contains edges of $\mathsf{OPT}$ that have $b(e) > \gamma_B$, and $\mathsf{OPT}^{-}$ contains edges of $\mathsf{OPT}$ that have $b(e) \le \gamma_B$. Similarly, let $\mathsf{OPT}(\gamma_B)$ be the optimal fractional matching on subgraph $G(\gamma_B)\subseteq G$, where $\gamma_B$ is the output threshold from the \textsc{Threshold} algorithm with graph $G$.  By definition of optimal matching, $u(\mathsf{OPT}^{-}) \le u(\mathsf{OPT}(\gamma_B))$. Moreover, for $\sfM$, the output matching from \textsc{Threshold} algorithm with graph $G$, we have 
$u(\sfM) \ge \frac{u(\mathsf{OPT}(\gamma_B))}{2}$, since $\sfM$ is a greedy matching on $G(\gamma)$ (subgraph with all edges having $b(e) \le \gamma_B$). Therefore, 
$u(\sfM) \ge \frac{u(\mathsf{OPT}^{-})}{2}$.

All edges $e = (\ell ,r) \in\mathsf{OPT}^{+}$, have $b(e) = \frac{c_{\ell}}{u(e)} > \gamma_B$. Thus, 
$u(\mathsf{OPT}^{+}) = \sum_{e =(\ell, r)   \in \mathsf{OPT}^{+}} x(\ell) u(e) < \frac{\sum_{e = (\ell ,r) \in \mathsf{OPT}^{+}}x(\ell) c_{\ell}}{\gamma_B}$, where $x(\ell)$ are fractional weights in the optimal solution. 
Moreover, the total budget constraint of $B$ ($\sum_{e = (\ell ,r) \in \mathsf{OPT}} x(\ell)c_{\ell} \le B$) implies that $u(\mathsf{OPT}^{+}) < \frac{B}{\gamma_B}$. Assuming that the budget constraint is tight with the \textsc{Threshold} algorithm ($\sum_{e\in \sfM}\gamma_B u(e)= B$), $u(\sfM) = \frac{B}{\gamma_B}$. Therefore, 
$u(\mathsf{OPT}^{+}) < u(\sfM)$. Combining this with $u(\sfM) \ge \frac{u(\mathsf{OPT}^{-})}{2}$, we have $u(\sfM^*) \le 3 u(\mathsf{OPT})$ as required.

If the budget constraint is not tight with the \textsc{Threshold} algorithm, then under our assumption that the maximum utility of any edge is $u_{max}$, and by the definition of \textsc{Threshold} algorithm that finds the largest feasible $\gamma$, the leftover budget $B-\sum_{e\in \sfM}\gamma_B u(e)$ is no more than $\gamma_B u_{max}$, and similar argument gives us that $u(\mathsf{OPT}) \le (3+o(1)) u(\sfM)$. 
\end{proof}
We next state a critical lemma for analyzing the proposed online version of \textsc{Threshold}, $\mathsf{ON}$. 

\begin{lemma}\label{lem:monotonegreedymatching}
Let $G = (L\cup R, E)$ and $F\subseteq G$, such that $F = (L\backslash L'\cup R, E')$, and the edge set $E'$ is such that all edges incident on left vertices in set $L'$ are removed simultaneously, while all edges incident on $L\backslash L'$ are retained as it is. Then $$u(\textsc{Greedy}(G)) \ge u(\textsc{Greedy}(F)).$$ Moreover $$u(\textsc{Greedy}(G(\gamma_1))) \ge u(\textsc{Greedy}(G(\gamma_2)))$$ for $\gamma_1 \ge \gamma_2$, and $u(\textsc{Greedy}(G(\gamma))) \ge u(\textsc{Greedy}(F(\gamma)))$.
\end{lemma}
\begin{proof} For arbitrary subgraph $F \subseteq  G$, $u(\text{\textsc{Greedy}}(G))$ may or may not be larger than $u(\text{\textsc{Greedy}}(F))$. However, when a left vertex is removed (by deleting all edges incident to it), the proof of claim $1$ follows standard procedure by showing that the weight of the edge incident on any right vertex in $\text{\textsc{Greedy}}(G)$ is at least as much as in  $\text{\textsc{Greedy}}(F)$. Detailed proof is omitted for lack of space. 
For the second  and third claim, note that an edge $e$ incident on left vertex $\ell$ is removed in $G(\gamma)$ compared to $G$, if 
$b(e) > \gamma$ or equivalently if $u(e) < \frac{c_{\ell}}{\gamma}$. 
Recall that the cost of any edge only depends on the index of its left vertex. Hence, if edge $e =(\ell, r)$ is removed from $G$ to obtain $G(\gamma)$, 
then all the edges $e'$ incident on $\ell$ with utility $u(e') < u(e) $ are also removed. So essentially, edges are removed monotonically from $G$ to produce $G(\gamma)$. So the proofs for the second and third claim follow similarly to the first.
%
%
\end{proof}
The importance of Lemma \ref{lem:monotonegreedymatching} is in showing that \textsc{Threshold} is solvable in polynomial time and the threshold $\gamma_B$ is monotonic. We prove the two claims as follows.

\begin{lemma}\label{lem:polytimecomplexity}
\textsc{Threshold} is solvable in polynomial time.
\end{lemma}
Algorithm \textsc{Threshold} involves finding a maximum in Step 4. We will show that one can use bisection to solve this maximization. We would like to note that if $u(\textsc{Greedy}(G(\gamma))) \ngtr u(\textsc{Greedy}(F(\gamma)))$, then finding this maximum is non-trivial. 
\begin{proof} From the definition of  Algorithm \textsc{Threshold} its clear that if any $\gamma \in \cA(G)$, then $\gamma_B \ge \gamma$. Hence the key step is to show that if any $\gamma \notin \cA(G)$, then $\gamma_B < \gamma$ which follows from the second claim of Lemma \ref{lem:monotonegreedymatching}, that $u(\textsc{Greedy}(G(\gamma_1))) \ge u(\textsc{Greedy}(G(\gamma_2)))$ for $\gamma_1 \ge \gamma_2$. Therefore, if for any $\gamma \notin \cA(G)$, then for any $\gamma' > \gamma$, $\gamma' \notin \cA(G)$. Hence we can use bisection to find the maximum.
\end{proof} 

The main and critical difference between the \textsc{Threshold} and  \textsc{UniformMechanism} \cite{goel2013matching} algorithm is the maximization step that ensures the following monotonicity property on $\gamma_B$ (Lemma \ref{lem:monotoneGamma}) that allows us to make the algorithm {\it online} using the \textsc{SampleandPrice} algorithm \cite{KorulaPal}.
\begin{lemma}\label{lem:monotoneGamma}
Let $G = (L\cup R, E)$ and $F = (L\backslash L'\cup R, E')$, where the edge set $E'$ is such that all edges incident on left vertices in set $L'$ are removed simultaneously, while all edges incident on $L\backslash L'$ are retained as it is. Then 
$\gamma_B(F) \ge \gamma_B(G)$.
\end{lemma}
\begin{proof}
From Lemma \ref{lem:monotonegreedymatching}, \begin{equation}\label{eq:FgGg}
u(\sfM(F(\gamma))) \le u(\text{\textsc{Greedy}}(G(\gamma))).
\end{equation}
Let the threshold and the matching obtained by running \textsc{Threshold} on $G$ with budget $B$ be $\gamma_B(G) = \gamma$, and $\sfM(G)$, respectively, where $\gamma \le \frac{B}{u(\text{\textsc{Greedy}}(G(\gamma)))}$. Now we consider $F(\gamma)$ as the input graph to the \textsc{Threshold} with same budget constraint $B$. Since $\gamma \le \frac{B}{u(\text{\textsc{Greedy}}(G(\gamma)))}$, from \eqref{eq:FgGg}, clearly, 
$\gamma \le \frac{B}{u(\sfM(F(\gamma))}$, and $\sum_{e\in \sfM(F(\gamma))}\gamma u(e)\leq B$. Therefore,  $\gamma \in \cA(F)$, which by definition of $\gamma_B(F)$ implies $\gamma_B(F) \ge \gamma$.
\end{proof}

We now describe our online algorithm $\mathsf{ON}$, that produces the matching $\sfM_{\mathsf{ON}}$ and associated payments for left vertices that are part of matching $\sfM_{\mathsf{ON}}$.
\begin{algorithm}
\caption{$\mathsf{ON}$ Algorithm}\label{alg:msandp}
\begin{algorithmic}[1]
\State {\bf Input:} $L$ set of left vertices/users that arrive sequentially, $R$ set of right vertices, Budget $B' = \frac{B}{\beta}$,  
\State $L_{1/2}$ = first half of left vertices $L$
\State Run \textsc{\textsc{Threshold}} on $G_{1/2}= (L_{1/2} \cup R, E_{1/2})$ with budget $B'$ to obtain $\gamma_{1/2}\triangleq\gamma_{B'}(G_{1/2})$ and matching $\sfM_{1/2}$
\For{each right vertex $r\in R$}
	\State Set value $v(r):=u(e)$ for $e = (\ell, r) \in \sfM_{1/2}$
	\EndFor
	\State \%Decision Phase
\State$\sfM_{\mathsf{ON}} =\emptyset$ 
\For{every new left vertex $\ell \in L \backslash L_{1/2}$},
	\State \%Pruning: 
	\State Delete all edges $e = (\ell, r), r\in R$ s.t. $b(e)>\gamma_{1/2}$
	\State Let $e^\star = \arg \max_{e=(\ell, r), r \in R, u(e) > v(r)} u(e)$ be the largest weight (utility) edge after pruning with weight larger than the value of the corresponding right vertex. 
	\State Let $e^\star$ be incident on right vertex $r^\star$ 	\If{$\sfM_{\mathsf{ON}} \cup  \{e^{\star}\}$ is a matching } 
	\State $\sfM_{\mathsf{ON}} = \sfM_{\mathsf{ON}}  \cup \{e^\star\}$ 
	\State Pay $p_{\ell} = \beta \gamma_{1/2} v(r^{\star})$ to user $\ell$ 
	\Else 
	\State Let $\ell$ be permanently unmatched
	\EndIf
\EndFor
\end{algorithmic}
\end{algorithm}
The idea behind $\mathsf{ON}$ is as follows:
\begin{itemize}
\item Do not match any of the first half of left vertices (called the observation phase), and only use them to run the offline \textsc{Threshold} algorithm and find the threshold $\gamma_{1/2}$ and the matching $\sfM_{1/2}$ with budget $B' = \frac{B}{\beta}$. Recall that $\beta > 1$, hence we are finding a conservative estimate for $\gamma_{1/2}$.
\item For any right vertex $r \in \sfM_{1/2}$, set its value $v(r)$ to be the weight (utility) of the matched edge in $\sfM_{1/2}$. We are assuming that $|L_{1/2}|$ is large enough compared to the number of right vertices and all right vertices are matched by \textsc{Threshold} in offline phase using the first half of the left vertices, i.e., $v(r) > 0, \ \forall \ r\in R$. 
\item In the decision phase, starting with the arrival of $|L/2|+1^{th}$ left vertex, delete all edges that have buck per bang $b(e)$ larger than $\gamma_{1/2}$. Among the surviving edges, match the left vertex to the right vertex with the largest weight (utility) that is higher than the value of the right vertex found from $\sfM_{1/2}$, if any. 
\item For each matched left vertex $\ell$, pay $p = \beta \gamma_{1/2} v(r)$, where $r$ is the right vertex to which the accepted edge from $\ell$ is matched. 
\end{itemize}

We now compute the expected utilities of matchings $\sfM_{1/2}$ and $\sfM_{\mathsf{ON}}$, where the expectation is over the uniformly random left vertex arrival sequences.
\begin{lemma}\label{lem:offguarantee} $\bbE\{u(\sfM_{1/2})\} \ge u(\mathsf{OPT}(\frac{B}{\beta}))/12$.
\end{lemma}
\begin{proof} Let $B' = \frac{B}{\beta}$. Let $G = (L \cup R, E)$ be the full graph, while $G_{1/2} = (L_{1/2}(\sigma) \cup R, E_{1/2}(\sigma))$, be the graph consisting of only the first half of left vertices that depends on arrival sequence $\sigma$.  


Since $G_{1/2}\subseteq G$, from Lemma \ref{lem:monotoneGamma}, we have $\gamma_{B'}(G_{1/2}) \ge \gamma_{B'}(G)$.
Thus, all the edges of $G(\gamma_{B'}(G))$ that are incident on left vertices $L_{1/2}(\sigma)$ are also present in the pruned graph $G_{1/2}(\gamma_{B'}(G_{1/2}))$. Let the greedy matching over the 'bigger' graph $G(\gamma_{B'}(G))$ be $\sfM(G)$. Let the subset of edges of $\sfM(G)$  that are incident on left vertices belonging to $L_{1/2}(\sigma)$ be $\sfM(G)_{\text{fh}}$.
Let the optimal fractional matching on $G_{1/2}(\gamma_{B'}(G_{1/2}))$ be $\mathsf{OPT}(B')_{1/2}$,
By definition, we have $$u(\mathsf{OPT}(B')_{1/2}) \ge u(\sfM(G)_{\text{fh}}).$$ Since we are considering the uniformly random arrival model for left vertices, i.e., $L_{1/2}$ is obtained by sampling each left vertex of $L$ with probability $\frac{1}{2}$, we have $\bbE\left\{u(\sfM(G)_{\text{fh}})\right\} \ge \frac{u(\sfM(G))}{2}$ and hence
\begin{equation}\label{eq:leftcont}
\bbE\left\{u(\mathsf{OPT}(B')_{1/2})\right\} \ge \frac{u(\sfM(G))}{2}.
\end{equation}

Moreover, since \textsc{Threshold} computes a greedy matching over $G_{1/2}(\gamma_{B'}(G_{1/2}))$, we have $u(\sfM_{1/2}) \ge \frac{u(\mathsf{OPT}(B')_{1/2})}{2}$ for any realization $\sigma$. From Lemma \ref{lem:umguarantee}, we already know that 
$u(\sfM(G)) \ge \frac{\mathsf{OPT}(B')}{3}$. Hence from \eqref{eq:leftcont}, we have $\bbE\{u(\sfM_{1/2})\} \ge u(\mathsf{OPT(B')})/12$.
\end{proof}

\begin{lemma}\label{lem:onguarantee} $\bbE\{u(\sfM_{\mathsf{ON}})\} \ge \bbE\{u(\sfM_{1/2})\}/2$.
\end{lemma}
\begin{proof} Consider the graph $G(\gamma_{B'}(G))$, where $\gamma_{B'}$ is the output threshold of THRESHOLD algorithm with budget $B'$ and graph  $G_{1/2}$. Then the setting of value for each right vertex, and the greedy selection of edges with weights larger than the value of right vertices  in the decision phase of $\mathsf{ON}$ is identical to running SAMPLEANDPRICE algorithm (Algorithm~\ref{alg:sampleandprice}) on graph $G(\gamma_{B'}(G))$ with  $p=\frac{|L_{1/2}|}{|L|} = \frac{1}{2}$, and hence it follows from Lemma 2.5 \cite{KorulaPal} that $\bbE\{u(\sfM_{\mathsf{ON}})\} \ge \bbE\{u(\sfM_{1/2})\}/2$. 
\end{proof}

\begin{algorithm}
\caption{$\mathsf{SAMPLEADPRICE}$ Algorithm}\label{alg:sampleandprice}
\begin{algorithmic}[1]
\State {\bf Input:} $G=(L\cup R,E)$ and $p \in [0,1]$
\State $k \leftarrow Binomial(|L|, p)$
\State Let $L'$ be the first $k$ vertices of $L$
\State $M_1 \leftarrow$ \textsc{Greedy}($G'$), with $G'=(L'\cup R, E')$
\For{each $r \in R$}
	\State Set $v(r) = u(e)$ of the edge $e$ incident to $r$ in $M_1$
\EndFor
\State $M \leftarrow \emptyset $
\For{each subsequent $\ell \in L \backslash L'$}
	\State Let $e = (\ell, r)$ be the highest-weight edge with $u(e) \ge v(r)$
	\If{$M\cup \{e\}$ is a matching}
		\State  $M = M\cup \{e\}$
	\EndIf
\EndFor
\end{algorithmic}
\end{algorithm}

The following theorem is the main result of the paper.
\begin{theorem}
Algorithm $\mathsf{ON}$ is $24\beta$-competitive, satisfies the budget constraint and is truthful.\end{theorem}
\begin{proof} The $24\beta$-competitiveness of $\mathsf{ON}$ follows from combining Lemma  \ref{lem:optscaling},\ref{lem:offguarantee}, and \ref{lem:onguarantee}. The budget feasibility and truthfulness are shown in Lemma \ref{lem:budfeas} and \ref{lem:ic}, respectively.
\end{proof}
\begin{lemma}\label{lem:budfeas}
Algorithm $\mathsf{ON}$ satisfies the payment budget constraint, and payment $p_{\ell} \ge c_{\ell}$ for any selected left vertex $\ell$.
\end{lemma}
\begin{proof}  Let $\gamma_{1/2} = \gamma_{B'}(G_{1/2})$ for simplicity. For a selected left vertex $\ell$, its buck per bang $\frac{c_{\ell}}{u(e)}\le \gamma_{1/2}$ and $u(e) > v(r)$, where $e = (\ell, r)$ is the selected edge. From the definition of $\beta$,  $\beta v(r) \ge u(e)$. Thus, $p_{\ell} = \gamma_{B'} \beta  v(r) \ge  \gamma_{1/2} u(e)$ and hence $p_{\ell} \ge c_{\ell}$.

From the definition of algorithm \textsc{Threshold}, we know that the output threshold $\gamma_{B'} \in \cA(G)$ for the offline phase (first half of left vertices) of $\mathsf{ON}$, and hence
\begin{equation}\label{eq:umbf}
\gamma_{1/2} \sum_{e \in \sfM_{1/2}}u(e) \le B'.
\end{equation}
Since the value of any right vertex $r$ is $v(r) = u(e)$ where $e = (\ell, r) \in  \sfM_{1/2}$. Therefore, from \eqref{eq:umbf}, we have that 
\begin{equation}\label{eq:sumvalue}
\sum_{r: e=(\ell, r) \in \sfM_{1/2}} v(r) \le \frac{B'}{\gamma_{1/2}}.
\end{equation}
Clearly, in the decision phase of $\textsf{ON}$, at most one left vertex is selected for each right vertex, and the payment made is $p_{\ell} = \beta \gamma_{B'} v(r)$ if $\ell$ is matched, and $p_{\ell} = 0$ otherwise. Thus, the total payment made 
$$\sum_{\ell, e=(\ell, r) \in \sfM_{\text{on}}} p_{\ell} \le \sum_{r, e=(\ell, r) \in \sfM_{1/2}}\beta \gamma_{1/2} v(r).$$
Thus, from \eqref{eq:sumvalue}, we get that 
$\sum_{\ell, e=(\ell, r) \in \sfM_{\text{on}}} p_{\ell} \le B$.
\vspace{-0.1in}
\end{proof}

Next, we show the most important property of $\mathsf{ON}$, its truthfulness. Towards that end, we will use the  Myerson's Theorem \cite{myerson1981optimal}.
\begin{theorem}\cite{myerson1981optimal}\label{Myerson_Theorem}
A reverse auction is truthful if and only if:
\begin{itemize}
\item The selection rule is monotone. If a user $\ell$ wins the auction by bidding $c_{\ell}$, it would also win the auction by bidding an amount $c_{\ell}'$, where $c_{\ell}' < c_{\ell}$.
\item Each winner is paid a critical amount. If a winning user submits a bid greater than this critical value, it will not get selected.
\end{itemize}
\end{theorem} 

\begin{lemma}\label{lem:ic}
$\mathsf{ON}$ is a truthful online algorithm.
\end{lemma}
\begin{proof} As stated before, the considered problem is a special case of a reverse auction. Thus to ensure that $\mathsf{ON}$ is truthful, we show that both the conditions of Theorem \ref{Myerson_Theorem} are satisfied. The monotone condition is easy to check, since in the decision phase, if any left vertex reduces its bid, then clearly its buck per bang $b(e)$ decreases, and hence it is still accepted if it was accepted before.

The second condition of payment being critical is also satisfied, shown as follows. Note that the payment made by $\mathsf{ON}$ to a selected left vertex $\ell$ is $p_{\ell} = \beta \gamma_{B'}(G_{1/2}) v(r), e =(\ell,r) \in \sfM_{on}$, where the right vertex index $r$ is such that utility $u(e), e=(\ell, r)$ is largest among the unmatched right vertices at the time of arrival of vertex $\ell$ that have an edge to left vertex  $\ell$, and $u(e) > v(r)$. 

Recall that  
$\frac{u_{max}}{u_{min}} \le \beta$, hence $\frac{u(e)}{v(r)} \le \beta$, where $e =(\ell,r)$. 
Now, if suppose the bid $c_{\ell}$ of left vertex $\ell$ is more than $p_{\ell}=\beta \gamma_{B'}(G_{1/2}) v(r)$, then since $u(e) \le \beta v(r)$, buck per bang of left vertex $\ell$, $c_{\ell} /u(e) >  \gamma_{B'}(G_{1/2})$. Moreover, since $u(e) > u(e')$ for all edges $e'$ incident on unmatched right vertices from $\ell$ at the arrival of left vertex $\ell$, we have that $c_{\ell}/u(e') >  \gamma_{B'}(G_{1/2})$. 
Thus, all edges out of left vertex $\ell$ incident on currently unmatched right vertices are removed in the pruning stage of the decision phase, and hence vertex $\ell$ cannot be selected.
\end{proof}

%
%
%
%

\bibliographystyle{unsrt}
\bibliography{onlined2d}

\end{document}